\documentclass[journal]{IEEEtran}
\usepackage{graphicx}
\usepackage{amsthm,amsfonts,amssymb,amsmath}
\usepackage{algorithm}
\usepackage{algpseudocode}
\usepackage{color}
\usepackage{bm,bbm}
\usepackage{multicol}
\newtheorem{prop}{Proposition}

\newtheorem{thm}{Theorem}
\newtheorem{rmk}{Remark}
\newtheorem{lemma}{Lemma}
\newtheorem{defi}{Definition}
\newtheorem{exe}{Example}
\newtheorem{cor}{Corollary}

\newcommand{\eff}{\text{eff}}
\newcommand{\vol}{\mbox{vol}}
\newcommand{\xx}{\mathbf{x}}
\newcommand{\yy}{\mathbf{y}}

\newcommand{\nozero}{\backslash\left\{0\right\}}

\begin{document}

\title{Random Ensembles of Lattices \\from Generalized Reductions}

\author{Antonio Campello, \textit{Member, IEEE} \thanks{
This work was supported in part by the S\~ao Paulo Research Foundation (fellowship 2014/20602-8).

\IEEEauthorrefmark{2}\textit{Current address}: Department of Electrical and Electronic Engineering, Imperial College London, South Kensington Campus, London SW7 2AZ, United Kingdom (e mail: a.campello@imperial.ac.uk). 
}}

\maketitle
\begin{abstract}
We propose a general framework to study constructions of Euclidean lattices from linear codes over finite fields. In particular, we prove general conditions for an ensemble constructed using linear codes to contain dense lattices (i.e., with packing density comparable to the Minkowski-Hlawka lower bound). Specializing to number field lattices, we obtain a number of interesting corollaries - for instance, the best known packing density of ideal lattices, and an elementary coding-theoretic construction of asymptotically dense Hurwitz lattices. All results are algorithmically effective, in the sense that, for any dimension, a finite family containing dense lattices is exhibited. For suitable constructions based on Craig's lattices, this family is smaller, in terms of alphabet-size, than previous ensembles in the literature. \\[1\baselineskip]
\textbf{Keywords: } Lattices, sphere packings, random codes, ideal lattices, codes over matrix rings \end{abstract}
\section{Introduction}
There has been a renewed interest in the search for new constructions of lattices from error-correcting codes due to their various recent applications, such as coding for fading wiretap channels \cite{KosiOngOggier}, Gaussian relay networks \cite{Adaptative}, compound fading channels \cite{Our} and index codes \cite{Index}, to name only a few. For the applications considered in these works, it is desirable to lift a code over a finite field into a lattice that possesses a rich algebraic structure, often inherited from the properties of number fields. In the present work we provide an unified analysis of these constructions and investigate ``random-coding''-like results for such lattices. Our focus is on the problem of finding dense structured lattice packings, although our techniques have a much broader scope of applications (as discussed in Section \ref{sec:conclusion}).

Indeed, finding the densest packing is a central subject in the Geometry of Numbers, with a variety of well-established connections to Coding Theory. Let $\Delta_n$ denote the best possible sphere packing density achievable by a Euclidean lattice of dimension $n$. The celebrated Minkowski-Hlawka theorem (e.g. \cite{Cassels,Leker}) gives the lower bound $\Delta_n \geq \zeta(n)/2^{n-1}$ for all $n\geq 2$, where $\zeta(n) = 1+1/2^n+1/3^n+\ldots$ is the Riemann zeta function. Up to very modest asymptotic improvements, this is still, to date, the best lower bound for the best packing density in high dimensions. 

Typical methods for establishing the theorem depend on the construction of random ensembles of lattices and on mean-value results \cite{Siegel,Rogers}. Rush \cite{Rush1989} and Loeliger \cite{Loeliger} obtained the lower bound from integer lattices constructed from linear codes in $\mathbb{F}_p^n$, in the limit when $p\to\infty$, with random-coding arguments. Improvements of the lower bound, in turn, strongly rely on additional (algebraic) structure. For instance, Vance \cite{Vance} showed that the best quaternionic lattice in dimension $m$ (with real equivalent in dimension $4m$) has density at least $3m \zeta(4m)/e 2^{4m-3}\leq \Delta_{4m}$. Using lattices built from cyclotomic number fields, Venkatesh \cite{Venkatesh} established the bound $\Delta_{2\varphi(m)} \geq m/2^{2\varphi(m)-1}$, where $\varphi(m)$ is Euler's totient function (since $m$ can grow as fast as a constant times $\varphi(m)\log \log \varphi(m)$ this  provides the first super-linear improvement). From another perspective, Gaborit and Z\'emor \cite{Zemor}, and Moustrou \cite{Moustrou} exploited additional coding-theoretic and algebraic structures to significantly reduce the family size of ensembles containing dense lattices. 
\\[1\baselineskip]
\noindent \textbf{Main Contributions.}
In this work we investigate general random lattices obtained from error correcting codes. The objective of this study is twofold: we provide unified analyses and coding-theoretic proofs of the aforementioned results, as well as a simple condition to verify if any new construction can be used to build ensembles containing dense lattices. 
We start from the fairly general definition of a \textit{reduction}, i.e. a mapping that takes a lattice into the space $\mathbb{F}_p^n$. For a general reduction we prove the following:
\begin{thm} Let $\phi_{p}: \Lambda \to \mathbb{F}_{p}^n$ be a family of reductions (surjective homomorphisms), where $\Lambda$ is a lattice of rank $m$ in the Euclidean space. Consider the ensemble
	$$\mathbb{L}_{p} = \left\{\beta \phi_p^{-1}(\mathcal{C}): \mathcal{C} \mbox{ is a } k-\mbox{dimensional code in } \mathbb{F}_p^n\right\},$$
	for an appropriate normalization factor $\beta$ so that all lattices have volume $V$. Denote by $\mathcal{N}_{\Lambda^{\prime}}(r) = \#( \mathcal{B}_r \cap \Lambda')$ the number of primitive points of $\Lambda$ inside a ball of radius $r$. If the first minimum of $\Lambda_{p} = \ker \phi_{p}$ satisfies
	\begin{equation}\liminf_{p\to \infty}\left( \frac{\lambda_1(\Lambda_{p})}{p^{n/m}}\right) > 0, \mbox{ then }
	\label{eq:nonDeg}
	\end{equation}
	$$\lim_{p\to\infty} E_{\mathbb{L}_{p}}[\mathcal{N}_{\Lambda^\prime}(r)] = (\zeta(m)V)^{-1} \text{\upshape vol } \mathcal{B}_r,$$
	where the average is with respect to the uniform distribution on $\mathbb{L}_{p}$.
	\label{thm:2}
\end{thm}

A slightly stronger version of the above result is precisely stated in Theorem \ref{thm:AverageEnsembleGeneral}. We shall refer to a family of reductions that satisfies condition \eqref{eq:nonDeg}  as \textit{non-degenerate}. Non-degeneracy is indeed a very mild condition, and is satisfied, for instance, if $\Lambda_p$ has non-vanishing Hermite parameter (e.g., $\Lambda_p = p\mathbb{Z}^n$). Non-degenerate constructions immediately yield lattices satisfying the Minkowski-Hlawka lower bound (see the discussion in the end of Section \ref{sec:prelim}). By choosing specific suitable families of non-degenerate reductions, we can further improve this density and obtain a number of interesting corollaries. We highlight one of them:

\begin{cor}\label{cor:betterCyclotomic} Let $\mathcal{O}_K$ be the ring of integers of a degree $n$ number field $K$ containing $r(K)$ roots of unity. For any integer $t\geq 2$, there exists an $\mathcal{O}_K$-lattice with dimension $t$ and packing density 
	$$\Delta \geq \frac{r(K) t \zeta(tn)(1-\varepsilon)}{e(1-e^{-t}) 2^{tn}},$$
	for any $\varepsilon > 0$.
\end{cor}
This proves for instance, the existence of ideal lattices in any dimension with density better, by a linear factor, than the Minkowski-Hlawka lower bound. This also recovers,
for $t=2$ and a judicious choice of number field and degree, the density in \cite[Thm. 1]{Venkatesh} and \cite[Thm. 2]{Moustrou}. By allowing reductions to codes over matrix rings (rather than the field $\mathbb{F}_p$), we provide, in Section \ref{sec:redRings}, a coding-theoretic proof of the existence of dense Hurwitz lattices, as in \cite{Vance}.

Here is how Theorem 1 may be interpreted: the density of the kernel (coarse) lattice $\Lambda_p$ is improved by ``adjoining'' a code $\mathcal{C}$ to it, through the reduction $\phi_p$. Now if $\Lambda_p$ itself has a reasonable density, we can improve it up to the Minkowski-Hlawka bound. Building on this idea, we show in Section \ref{sec:conc} that if we start from a suitable reduction so that the base (fine) lattice $\Lambda$ is not so thick (in terms of its \textit{covering density}) and such that the kernels are not so sparse (in terms of \textit{packing density}), we can bound the required size of $p$ to be within a finite range. For instance, we show that by starting from the family of Craig's lattices \cite{Splag}, we can build dense lattices from codes with alphabet-size $p = O((n\sqrt{\log n})^{1+\nu})$, where $\nu > 0$ is any (small) positive constant. This improves significantly the size of codes required by usual constructions \cite{Rush1989} \cite{Loeliger} (where $p$ grows at least as $\omega(n^{3/2})$). As observed in \cite[pp. 18-19]{Splag}, the works of Rush (and Loeliger) already significantly reduce the family sizes of typical proofs of the Minkowski-Hlawka lower bound. It is worth mentioning that, in terms of absolute family size, the best result is achieved by \cite{Zemor} by restricting the average to double-circulant codes or \cite{Moustrou} using cyclotomic lattices. For instance, while the logarithm of the search space for Craig's lattices reductions has size $n^2 \log \log n$, the family based on double-circulant codes has size $n \log n$ (we refer the reader to Table \ref{tab:comparison} for more details). We leave it as an open question to whether coupling a good reduction with a smaller family of codes can further reduce the overall search space.
\\[1\baselineskip]
\noindent \textbf{Organization. } This work is organized as follows. In Section \ref{sec:prelim} we describe some basic definitions and notation. In Section \ref{sec:genRed} we establish our main result on general reductions and several corollaries. In Section \ref{sec:numbFields} we consider reductions induced by quotients of ideals in the ring of integers of a number field, proving the main corollaries.  In Section \ref{sec:redRings} we construct random Hurwitz and Lipschitz lattices from codes over matrix rings. In Section \ref{sec:conc}, we discuss an ``algorithmic'' version of the main theorem and draw the final conclusions.




\section{Preliminaries and Notation}
\label{sec:prelim}
The Euclidean norm of $\xx \in \mathbb{R}^m$ is denoted by $\left\|\xx \right\| = (x_1^2+\ldots+x_m^2)^{1/2}$. The ball of radius $r$ in $\mathbb{R}^m$ is denoted by $\mathcal{B}_r = \left\{\xx \in \mathbb{R}^m: \left\|\xx \right\| \leq r \right\}$.
A \textit{lattice} $\Lambda$ is a discrete additive subgroup of $\mathbb{R}^m$. Denote by $\mbox{span } \Lambda$ the minimal subspace of $\mathbb{R}^m$ that contains $\Lambda$.  The \textit{rank} of $\Lambda$ is defined to be the dimension of $\mbox{span } \Lambda$. The quotient $(\mbox{span } \Lambda)/ \Lambda$ is compact, and its volume, denoted by $V(\Lambda)$, is said to be the \textit{volume of }$\Lambda$. The \textit{first minimum} $\lambda_1(\Lambda)$ of $\Lambda$ is the shortest Euclidean norm of non-zero vectors in $\Lambda$.  In general, the $i$-th minima of $\Lambda$ are defined as
\begin{equation*}
\lambda_i(\Lambda) = \min\left\{r : \dim\left( \mbox{span}\left\{\mathcal{B}_r \cap \Lambda\right\}\right) = i\right\}.
\end{equation*} The \textit{packing density} of a rank $m$ lattice $\Lambda$ is defined as $$\Delta(\Lambda) = \frac{\mbox{vol }\mathcal{B}_{\lambda_1/2}}{V(\Lambda)}.$$
We say that a point in $\xx \in \Lambda$ is \textit{primitive} if the intersection between $\Lambda$ and the open line segment $\left\{\alpha \xx : \alpha \in (0,1)\right\}$ is the empty set. The set of all primitive points in $\Lambda$ is denoted by $\Lambda^\prime$.

Theorem \ref{thm:2} implies the Minkowski-Hlawka lower bound in the following  fashion. From the average result, it follows that it must exist at least one $\Lambda \in \mathbb{L}$ such that $\mathcal{N}_{\Lambda}(r) \leq (\zeta(m)V)^{-1} \text{vol } \mathcal{B}_r^m$. Now if we force the right-hand side to be equal to $2(1-\varepsilon)$, for some small $\varepsilon >0$, then, since a lattice has at least two minimum vectors, we must have $\mathcal{N}_{\Lambda}(r) = 0$. Therefore $\Lambda$ can pack balls of radius $r/2$; rearanging the terms gives us, up to $\varepsilon$, the Minkowski-Hlawka bound. If $\Lambda$ is a lattice with guaranteed number of minimum vectors (say, $L$) we can, by similar arguments, achieve density $L(1-\varepsilon)/2^{m}$.

A $k$-dimensional vector subspace $\mathcal{C} \subset \mathbb{F}_p^n$ is called a (linear) \textit{code} with parameters $(n,k,p)$ (or simply an $(n,k,p)$-code).

Throughout the paper we use the standard ``big-O'', ``big-omega'' and ``little-omega''  notations, e.g. $f(x) = \Omega(g(x))$ if $\lim \sup_{x\to \infty} |g(x)/f(x)| < +\infty$, and $f(x) = \omega(g(x))$ if $\lim_{x \to \infty} g(x)/f(x) = 0$.
\section{Generalized Reductions}
\label{sec:genRed}
From now on, let $\Lambda$ be a rank $m$ lattice and let $n \leq m$ be an integer. 
\begin{defi} Let $\phi_p:\Lambda\to \mathbb{F}_p^n$ be a surjective homomorphism. Given a linear code $C$, its associated lattice via $\phi_p$ is defined as $\Lambda_p(C) \triangleq \phi_p^{-1}(C)$.
	\label{def:ConsGen}
\end{defi} 
A surjective homomorphism as in the above definition will, from now on, be called a \textit{reduction}. We shall see that $\Lambda_p(C)$ is indeed a lattice of rank $m$. First observe that $\Lambda_p(C)$ is a subgroup of $\Lambda_p(\mathbb{F}_p^n)=\Lambda$. Since the quotient $\Lambda/\text{ker}(\phi_p) \simeq \mathbb{F}_p^n$ is finite, $\text{ker}(\phi_p)=\Lambda_p(\left\{\mathbf{0}\right\})\triangleq\Lambda_p$ is a sub-lattice of $\Lambda$, of rank $m$. From the inclusion $\Lambda_p \subset \Lambda_p(C) \subset \Lambda$, we conclude that the three lattices have the same rank. Moreover, $\Lambda_p(C)/\Lambda_p \simeq C$, and therefore $V( \Lambda_p(C)) = \left|C\right|^{-1} p^{n} V( \Lambda)$.
\begin{rmk} There is an off-topic connection between Definition \ref{def:ConsGen} and combinatorial tilings. If in addition to being surjective, the reduction $\phi_p$ is a bijection when restricted to a set $\mathcal{P} \subset \Lambda$ of cardinality $p^n$, then $\mathcal{P}$ tiles $\Lambda$ by translations of vectors of $\Lambda_p$.
\end{rmk}
This framework contains classical Construction A \cite{Splag}, \cite{Rush1989}, the constructions in \cite{KosiOngOggier}, and \cite{ConsAOggier}. We derive sufficient conditions for this general construction to admit a Minkowski-Hlawka theorem. Set $\beta= {V^{1/m}}/{(p^{n-k}V(\Lambda)^{1/m})}$ and let
\begin{equation}
\mathbb{L}_p=\left\{ \beta \Lambda_p(C) : C \mbox{ is an } (n,k,p)-\mbox{code}\right\}
\label{eq:ensemble}
\end{equation}
be the ensemble of all lattices associated to codes of dimension $k$, normalized to volume $V$. Suppose that a lattice in $\mathbb{L}_p$ is picked at random by choosing $C$ uniformly. We shall prove a generalized version of the Minkowski-Hlawka theorem for $\mathbb{L}_p$. Instead of functions with bounded support, we will consider a wider class of functions. Let $W=\mbox{span}(\Lambda)$ be the minimal subspace of $\mathbb{R}^n$ containing $\Lambda$ (therefore, $\mbox{dim}( \mbox{span}(\Lambda_p(C))) = m$). 
\begin{defi} Let $f:W\to \mathbb{R}$ be a Riemann-integrable function. We say that $f$ is \textit{semi-admissible} if 
	\begin{equation}
	|f(\xx)| \leq \frac{b}{(1+\left\| \xx \right\|)^{m+\delta}}, \forall \xx \in W
	\end{equation}
	where $b > 0$ and $\delta > 0$ are positive constants.
\end{defi}
Any bounded integrable function with compact support is semi-admissible. Of particular interest are indicator functions of bounded convex sets. Notice that, for any semi-admissible function and a rank-$m$ lattice $\Lambda$ in $W$,
$$\sum_{\xx \in \Lambda} {f(\xx)} < +\infty.$$

\begin{rmk} If $f$ and its Fourier transform $\hat{f}$ are semi-admissible, then $f$ is said to be admissible. In this paper we will not be concerned about admissible functions, which play an important role in the development of the so-called linear programming bounds for packings.
\end{rmk}

\begin{thm} Let $(p_j)_{j=1}^{\infty}$ be an increasing sequence of prime numbers such that there exist reductions $\phi_{p_j}:\Lambda\to\mathbb{F}_{p_j}^n$ and let $f:W\to\mathbb{R}$ be a semi-admissible function. If the first minimum of $\Lambda_{p_j}=\Lambda_{p_j}(\left\{0\right\})$ satisfies 
	$$\lambda_1(\Lambda_{p_j}) \geq c p_{j}^{\frac{n-k}{m}+\alpha},$$ for some constant $c,\alpha>0$, then
	
\noindent (i) \begin{equation}
			\lim_{p_j\to\infty} \mathbb{E}_{\mathbb{L}_{p_j}}\left[\sum_{\xx \in \beta \Lambda_{p_j}(C)\nozero} f(\xx) \right] = V^{-1}\int_{W} f(\xx)\text{d}\xx,
			\label{eq:averageBehaviorNoTheta}
			\end{equation}
(ii) \begin{equation}
	\lim_{p_j\to\infty} \mathbb{E}_{\mathbb{L}_{p_j}}\left[\sum_{\xx \in \beta \Lambda_{p_j}^{\prime}(C)} f(\xx) \right] = (\zeta(m)V)^{-1}\int_{W} f(\xx)\text{d}\xx,
	\label{eq:averageBehavior}
	\end{equation}
	\label{thm:AverageEnsembleGeneral}where the averages are taken over all $\beta\Lambda_{p_j}(\mathcal{C})$ in the ensemble $\mathbb{L}_{p_j}$ (Equation \eqref{eq:ensemble}).
\end{thm}
\begin{proof} We will prove the ``refined'' statement (ii). The proof of (i) is similar, except for the last step. Recall that the set of $\mathcal{C}_{n,k}$ of all $(n,k)$-codes satisfies Loeliger's balancedness equation \cite{Loeliger} 
	\begin{equation}
	\frac{1}{| \mathcal{C}_{n,k}|} \sum_{{C} \in \mathcal{C}_{n,k}} \sum_{c \in {C} \backslash \left\{\mathbf{0}\right\}}  g(c)= \frac{p_j^{k}-1}{p_j^{n}-1} \sum_{v \in  \mathbb{F}_{p_j}^n \backslash \left\{\mathbf{0}\right\}}g(v),
	\label{eq:LoeligerLemma}
	\end{equation}
	for a function
	$g:\mathbb{F}_p^n \to \mathbb{R}$. Now for $f:W\to \mathbb{R}$, \begin{equation}
	\begin{split} &\mathbb{E}\left[\sum_{\xx \in \beta\Lambda_{p_j}^{\prime}(C)} f(\xx) \right] = \mathbb{E}\left[\sum_{{\xx \in \beta\Lambda_{p_j}^{\prime}(C)}\above 0pt {\phi_{p_j}(\xx/\beta)=0}} f( \xx) \right] \\ +& \mathbb{E}\left[\sum_{{\xx \in \beta\Lambda_{p_j}^\prime(C)}\above 0pt {\phi_{p_j}(\xx/\beta)\neq0}} f(\xx) \right]. \\
	\end{split}
	\label{eq:splitAverage}
	\end{equation}
	From the assumption on $f$, 
	\begin{equation*}\begin{split} &\left|\sum_{{\xx \in \Lambda_{p_j}^{\prime}(C)}\above 0pt {\phi_{p_j}(\xx)=0}} f(\beta \xx)\right| =\left|\sum_{{\xx \in \Lambda_{p_j}^{\prime}}} f(\beta \xx)\right| \\ 
	&\leq \sum_{{\xx \in \Lambda_{p_j}\backslash\left\{\mathbf{0}\right\}}} \frac{b}{(1+\left\| \beta \xx \right\|)^{m+\delta}}.
	\end{split}
	\end{equation*}
	Since the lattice $\Lambda_p(C)$ has rank $m \geq 1$, the series on the right-hand-side of the above inequality is absolutely convergent for any $p_j$. Moreover, since, by assumption $$\left\| \beta \mathbf{x} \right\| \geq \beta \lambda_1(\Lambda_{p_j})\geq c V^{1/m} (V(\Lambda))^{-1/m} p_j^{\alpha} \to \infty,$$ each individual term of  the last sum tends to zero, as $p_j \to \infty$ and, by dominated convergence, the sum tends to zero. Let $\gamma = (p_j^k-1)/(p_j^n-1)$. For the second term of Equation \eqref{eq:splitAverage}, we have:
	\begin{equation*}
	\begin{split} &\mathbb{E}\left[\sum_{{\xx \in \beta\Lambda_{p_j}^{\prime}(C)}\above 0pt {\phi_{p_j}(\xx)\neq0}} f(\beta \xx) \right]  \stackrel{(a)}{=} \gamma \sum_{\xx \in \Lambda^{\prime}}f(\beta\xx)  \\
&	\stackrel{(b)}{=}  \sum_{r = 1}^{\infty} \frac{\mu(r)}{r^m} \sum_{\xx \in \Lambda\nozero}r^{m} \gamma f(r\beta \xx)
	\end{split}
	\end{equation*}
 where $\mu$ denotes the M\"obius function (see, e.g. \cite[Sec. VI. 3.2]{Cassels}). In the above, (a) follows from \eqref{eq:LoeligerLemma} and (b) is the M\"obius function inversion formula (see, e.g., \cite[Sec. VI. 3.2]{Cassels}). The theorem follows by using the property 
 $$\sum_{r=1}^{\infty}\frac{\mu(r)}{r^m} = \frac{1}{\zeta(m)}$$
 and observing that for the inner sum satisfies
$$\lim_{p_j \to \infty} \sum_{\xx \in \Lambda\nozero}r^{m} \gamma f(r\beta \xx) = (V^{-1}\zeta(m))\int_{W} f(\xx) \mbox{d} \xx,$$
by the definition of Riemann integral. Exchanging the limit and the sum is justified by dominated convergence, given the condition on $f$.

\end{proof}


\begin{exe} If $\Lambda=\mathbb{Z}^n$ and $\phi_p$ is the reduction modulo $p$, we obtain mod-$p$ lattices as in \cite{Loeliger}. It is clear that $\Lambda_p = p\mathbb{Z}^n$ satisfies the hypothesis of Theorem \ref{thm:AverageEnsembleGeneral}, with $m=n$ and $\alpha = k/n$. This implies Theorem 1 of \cite{Loeliger}.
\end{exe}

\begin{defi} A sequence of surjective homomorphisms $(\phi_j)_{j=1}^{\infty}	$, $\phi_j:\Lambda \to \mathbb{F}_{p_j}^n$ is said to be \textit{non-degenerate} if 
	$$\lambda(\Lambda_{p_j}) \geq c p_{j}^{\frac{n}{m}},$$
	for some constants $c > 0$. Similarly, the sequence of associated ensembles (Equation \ref{eq:ensemble}) are said to be non-degenerate.
\end{defi}
It follows that if the reductions are non-degenerate, the associated ensemble admits the Minkowski-Hlawka theorem. 

\begin{exe}[``Natural reduction''] If $m=n$, the natural reduction to $\mathbb{F}_p^n$ is as follows. Given a basis $\xx_1,\ldots,\xx_n$ for $\Lambda$, take $\phi_p$ to be the linear map defined by $\phi_p(\xx_i) = \bm{e}_i \in \mathbb{F}_p^n$, where $\mathbf{e}_i$ the $i$-th canonical vector $(0,\ldots,0,1,0,\ldots,0)$. It is clear that $\phi$ is surjective and $\ker \phi_p = p \Lambda$, therefore the associated sequence of reductions is non-degenerate. This provides a systematic way of constructing good sublattices of a given lattice.
	\label{ex:naturalReduction}
\end{exe}
Taking $f(x)$ to be the indicator function of a ball in part (ii) of Theorem \ref{thm:AverageEnsembleGeneral}, we recover Theorem \ref{thm:2}. Another function of interest is $f(x)= e^{-\tau \left\|x\right\|^2}$ for $\tau > 0$, yielding the \textit{theta series}
$$\Theta_{\Lambda}(\tau) = \sum_{x \in \Lambda} e^{-\tau \left\|x\right\|^2}.$$
A corollary of part (i) of the theorem is the following:
\begin{cor} The average theta series of a sequence of non-degenerate ensemble satisfies
	\begin{equation}
	\lim_{p_j \to \infty} E_{\mathbb{L}_{p_j}}\left[\Theta_{\Lambda}(\tau)\right] = V^{-1} \left(\frac{\pi}{\tau}\right)^{m/2} + 1.
	\end{equation}
	\label{cor:thetaSeries}
\end{cor}
Corollary \ref{cor:thetaSeries} can, for instance, be applied to the construction of sufficiently flat Gaussian measures for secure communications (cf \cite{LLBS_12}).

\begin{rmk} The condition for non-degeneracy can be re-written as
	$$\liminf_{p_j\to\infty} \gamma(\Lambda_{p_j})  > 0,$$
	where $\gamma(\Lambda) = \lambda(\Lambda)/V(\Lambda)^{1/m}$ is the Hermite parameter of $\Lambda$. In other words, non-degeneracy is equivalent to non-vanishing Hermite parameter of the sequence of kernel lattices. 
\end{rmk}

%
%
We close this section with another consequence of Theorem \ref{thm:AverageEnsembleGeneral}. We shall refer to each ratio
\begin{equation}
\Delta_i(\Lambda) = \frac{\mbox{vol } \mathcal{B}_{\lambda_i/2}}{V(\Lambda)}, i = 1, \ldots, m,
\label{eq:successiveDensities}
\end{equation}
as the $i$-th successive density of a lattice $\Lambda$.
For a sequence of non-degenerate ensemble, put $\mathbb{L} = \bigcup_{j=1}^{\infty}\mathbb{L}_{p_j}$.
\begin{cor}
	For any $\varepsilon > 0$, there exists $\Lambda_{p_j} \in \mathbb{L}$ such that 
	\begin{equation}
	\prod_{i=1}^m \Delta_i(\Lambda)^{1/m} \geq \frac{2 m \zeta(m)(1-\varepsilon)}{e(1-e^{-m})}.
	\end{equation}
\end{cor}
\begin{proof}
	The proof follows from a method of Rogers \cite{Rogers}, choosing $f(x)$ appropriately in the Minkowski-Hlawka theorem. We shall give a complete proof in the next section, in the context of $\mathcal{O}_K$-lattices.
\end{proof}

\section{Constructions From Number Fields}
\label{sec:numbFields}
From now on we consider constructions of random ensembles based on algebraic number theory. We refer the reader to \cite{Fermat} and \cite{Notes} for an introduction to the theory, as well as undefined notation.

Let $K/\mathbb{Q}$ be a number field with degree $n$ and signature $(r_1,r_2)$. Denote its real embeddings by $\sigma_1,\ldots,\sigma_{r_1}$ and their pairs of complex embeddings by $$\sigma_{r_1+1}, \overline{\sigma_{r_1+1}},\ldots,\sigma_{r_1+r_2+1},\overline{\sigma_{r_1+r_2+1}}.$$ Let $\mathcal{O}_K$ be the ring of integers of $K$ and $\mathcal{I}\subset \mathcal{O}_K$ be an ideal. An ideal can be identified with a real lattice of dimension $(r_1+2r_2)$ via the canonical embedding 
$$\sigma : \mathcal{O}_K \to \mathbb{R}^{r_1+2r_2}$$
\begin{equation*}
\begin{split}\sigma(x) = (\sigma_1(x),\ldots,\sigma_{r_1}(x),&\Re \sigma_{r_1+1}(x),\ldots \Re \sigma_{r_1+r_2+1}(x),\\ &\Im \sigma_{r_1+1}(x),\ldots,\Im \sigma_{r_1+r_2+1}(x)).
\end{split}
\end{equation*}
Lattices constructed from the embedding of ideals $\mathcal{I} \subset \mathcal{O}_K$ are called \textit{ideal lattices}, and appear in the study of modular forms, coding theory, and cryptography. In this section we study the Minkowski Hlawka theorem for $\mathcal{O}_K$-lattices and related structures.

Let $E=K\otimes_{\mathbb{Q}} \mathbb{R} $ be the Euclidean space generated by $K$. An $\mathcal{O}_K$-lattice is a free $\mathcal{O}_K$ sub-module of $E^t$, for some $t>0$. In particular, an $\mathcal{O}_K$ lattice is closed under multiplication by elements of $\mathcal{O}_K$. The Euclidean norm in $E$ is induced by the trace form. Notice that $K$ is naturally embedded in $E$. In the cases when $K$ is either totally real or a totally imaginary extension of a real number field (CM-field) some notational simplifications can be made. For instance, we can write the trace form as
$$\text{tr}(x\overline{y}) = \sigma_1(x)\overline{\sigma_1(x)} + \ldots + \sigma_n(x)\overline{\sigma_n(x)}.$$

We discuss the average behavior of a general reduction from algebraic number theory \cite{KosiOngOggier,Our,Adaptative}, defined in the sequel. A prime $p$ is said to split completely if $p\mathcal{O}_K$ can be factored into the product of prime ideals $\mathfrak{p}_1 \mathfrak{p}_2 \cdots \mathfrak{p}_n$.


\begin{defi} Let $p$ be a prime that splits completely, and $\mathfrak{p}$ an ideal above $p$. Consider $\pi: \mathcal{O}_K \to \mathcal{O}_K/\mathfrak{p} \simeq \mathbb{F}_p$ a projection onto $\mathfrak{p}$ and $\sigma$ the canonical embedding. Let $\Lambda=\sigma(\mathcal{O}_K)^t$ (the canonical embedding is applied componentwise). Take
	$$\phi_p:\Lambda\to\mathbb{F}_p^t$$
	$$\phi_p(\sigma(x_1,\ldots,x_t))=(\pi(x_1),\ldots,\pi(x_t))$$
	and define $\Lambda_p(\mathcal{C})=\phi_p^{-1}(C) \subset \mathbb{R}^{nt}$.
	\label{def:con2}
\end{defi}

%

\begin{lemma} The ensemble induced by Definition \ref{def:con2} is non-degenerate.
\end{lemma}
\begin{proof} 
	The minimum algebraic norm of an element of $\mathfrak{p}$ is greater or equal than $p$. Hence $\Lambda_p=\sigma(\mathfrak{p})^t$ has minimum norm at least $\sqrt{n} p^{1/n}$, finishing the proof.
\end{proof}

A very important caveat to the previous lemma is the fact that there must exist an infinite number of primes $p$ such that the construction above is possible. This follows from Chebotarev's density theorem (e.g. \cite{Algebraic1}{ Cor. 13.6, p. 547}), which implies that the natural density $\delta$ of primes that split completely in $K$ is positive (indeed, one has $0 < \delta \leq 1/n!$).

\begin{rmk}
	Very similarly, it is possible to prove that the constructions in \cite{ConsAOggier} are non-degenerate.
\end{rmk}

Suppose that $K$ contains $r(K)$ roots of unity. Let $\mu$ be a root of unity contained in $\mathcal{O}_K$. It follows that $\left\| \sigma(\mu \xx )\right\| = \left\| \sigma(\xx )\right\|$ (e.g. \cite[Lem. 3.1]{Autissier}). Therefore, each $\Lambda$ constructed as in Definition \ref{def:con2} contains at least $r(K)$ minimal vectors and we automatically obtain the density $r(K)(1-\varepsilon)/2^{nt}$ (this argument was used by Venkatesh \cite{Venkatesh} to prove that cyclotomic lattices ($\mathbb{Z}[\mu]$-lattices where $\mu$ is a root of unity) achieve density $m(1-\varepsilon)/2^{2\phi(m)}$). For $t$-dimensional $\mathcal{O}_K$ lattices, however, there is a loss of a linear factor of $t$ in the enumerator. Nevertheless, we can improve the following density up to that of Corollary \ref{cor:betterCyclotomic}, using a method by Rogers \cite{Rogers}, recently employed in \cite{Vance} to quaternionic lattices. The basic idea is to apply to Theorem \ref{thm:AverageEnsembleGeneral}, rather than indicator function of a ball, to a bounded-support function that allows us to analyze the generalized densities of the ensemble. After ensuring the existence of a lattice with good generalized densities, it is possible to apply standard linear transformations to such a lattice (see e.g. \cite[Thm. 2.2]{Vance}) in order in order to transform it into a lattice with good \textit{packing density}.

\mbox{} \\
\textit{Proof of Corollary \ref{cor:betterCyclotomic}:} . For $\Lambda_0 \subset \mathcal{O}_K^t$ let the $i$-th successive minima of $\Lambda_0$  (over $K$) be the smallest $i$-such that the ball $\mathcal{B}_r$ contains the canonical embedding of $i$ linearly independent vectors (over $K$). More formally
\begin{equation}
\begin{split}
\lambda_i^K(\Lambda_0) = \min \left\{r > 0: \dim \mbox{span}_{K} \left(\sigma^{-1} \left(\sigma(\Lambda_0) \cap \mathcal{B}_r \right)\right)=i \right\}.
\end{split}
\end{equation}
Notice that $\lambda_1^K(\Lambda_0) = \lambda_1(\sigma(\Lambda_0))$ and, in general $\lambda_i^K(\Lambda_0) \geq \lambda_i(\sigma(\Lambda_0))$. Also, if $\xx_1,\ldots,\xx_t$ are linearly independent over $K$ and achieve the sucessive minima of $\Lambda_0$, then the embeddings $\sigma(\xx_1),\ldots, \sigma(\xx_t)$ are linearly independent and primitive in $\sigma(\Lambda_0) \subset \mathbb{R}^{nt}$. Now let $f:\mathbb{R}^{nt} \to \mathbb{R}$ be the following function with limited support:
\begin{equation}
f(\yy)=\left\{ \begin{array}{cc} 
1/n & \mbox{if }\left\|\yy \right\| \leq re^{(1-t)/tn} \\
\frac{1}{nt} - \log\left(\frac{\left\|\yy\right\|}{r}\right) & \mbox{if } re^{(1-t)/tn}\leq \left\|\yy \right\| \leq re^{1/tn}\\
0 & \mbox{otherwise.}\end{array}\right. 
\end{equation} 
We have 
\begin{equation*}
\int_{\mathbb{R}^{nt}} f(\yy) \mbox{d}\yy = \frac{e(1-e^{-t}) r^{nt} \mbox{vol }\mathcal{B}_1}{nt}.
\end{equation*}
Choose $r$ such that the right-hand side of this last equation is equal to $r(K)V\zeta(nt)(1-\varepsilon)/n$ for a small $\varepsilon < 1$.
Let $\phi_p$ be as in Definition \ref{def:ConsGen} and $\mathbb{L}_p$ its induced ensemble
\begin{equation}
\mathbb{L}_p=\left\{ \beta \Lambda_p(C) : C \mbox{ is an } (n,k,p)-\mbox{code}\right\}
\end{equation}
as in Equation \eqref{eq:ensemble}. According to Theorem \ref{thm:AverageEnsembleGeneral}, it is possible to find $\Lambda_1 = \beta \sigma(\Lambda_0) \in \mathbb{L}_{p}$ of volume $V$ such that, for $p$ sufficiently large
\begin{equation}
\sum_{\yy\in\Lambda_1^{\prime}} f(\yy) \leq (1-\varepsilon) \frac{r(K)}{n} < \frac{r(K)}{n},
\end{equation}
Let $v_1, \ldots, v_t$ be linearly independent vectors in $\Lambda_0$ achieving the successive minima, $\left\|\beta \sigma(v_i)\right\|=\lambda_i^K(\Lambda_0)$. We have 
$$\sum_{\yy\in\Lambda_1^{\prime}} f(\yy) \geq \sum_{i=1}^t \sum_{\mu} f(\beta \sigma(\mu v_i))  = r(K) \sum_{i=1}^t f(\beta \sigma(v_i)),$$
where the sum with subscript $\mu$ is over all roots of unity in $K$.
From this we conclude that, for all $i$, $\beta \lambda_i^K(\Lambda_0) \geq r e^{1/n-1}$ and
$$\frac{1}{n} - \log\left(\frac{\beta^t 
\prod \lambda_i^K(\Lambda_0)}{r^t}\right) < \frac{1}{n}.$$
Therefore, for the $t$ successive densities (Eq. \eqref{eq:successiveDensities}):

\begin{equation}\begin{split}
\left(\prod_{i=1}^{t} \Delta_i^K\right)^{1/t} &= \prod_{i=1}^t \left( \frac{\mbox{vol}(\mathcal{B}_{\lambda_i^K/2})}{V} \right)^{1/t} \\ &\geq \frac{ r(K) t \zeta(nt) (1-\varepsilon)}{e(1-e^{-t})2^{nt}}.
\end{split}
\end{equation}
But in this case, we can find $\tilde{\Lambda}$ whose packing density (or $\Delta_1^K$) is greater or equal than $\frac{ r(K) t \zeta(nt)}{e(1-e^{-t})2^{nt}(1-\varepsilon)}$ (e.g. \cite[Thm. 2.2]{Vance}). 
\qed
%
%
%
\section{Balanced Sets of Codes over Matrix Rings}
\label{sec:redRings}
In some contexts, the ``natural'' underlying alphabet in the reduction $\phi_p$ is, rather than the field $\mathbb{F}_p$, the ring  $\mathcal{M}_n(\mathbb{F}_p)$ of $n \times n$ matrices with entries in $\mathbb{F}_p$. Although we can identify $\mathcal{M}_n(\mathbb{F}_p)$ with $\mathbb{F}_p^{n^2}$, the identification does not carry enough algebraic structure for our purposes. For instance, we cannot guarantee that the constructed lattices are closed under multiplication by units, which is crucial in order to obtain the full density improvements of these lattices, as in \cite{Vance}.  For this reason, we study in this section a version of Theorem \ref{thm:AverageEnsembleGeneral} for codes over matrix rings.
\subsection{An Averaging Bound for Codes over Rings}
Let $\mathcal{R}$ be a finite ring and $\mathcal{R}^{*}$ its units. Denote by $(\mathcal{R}^n)^{*}$ the set of vectors in $\mathcal{R}^n$ such that at least one coordinate is a unit. A linear code in $\mathcal{C} \subset \mathbb{R}^n$ is a \textit{free}\footnote{This may differ from the literature, where a linear code over a ring is simply an additive subgroup of $\mathcal{R}^n$. The requirement that a linear code is a \textit{free} module is necessary for Lemma \ref{lem:balancedRings} to hold.} $\mathcal{R}$-submodule of $\mathcal{R}^n$ (with the natural scalar multiplication). Following \cite{Loeliger}, we define balanced sets of codes as follows.

\begin{defi} Consider a non-empty set of codes $\mathcal{C}_b$ of same cardinality. We say that $\mathcal{C}_b$ is \textit{balanced} if any $\xx \in (\mathcal{R}^n)^{*}$ is contained in the same number of codes (say, $L$) of $\mathcal{C}_b$. 
\end{defi}
%
Let $M$ be the cardinality of a code in $\mathcal{C}_b$. From a counting argument, one can see that $M |\mathcal{C}_b| \geq L |(\mathcal{R}^n)^{*}|$. The following lemma shows how to bound averages of functions in $(\mathcal{R}^n)^*$.

\begin{lemma} Let $g : \mathcal{R}^n \to \mathbb{R}^+$ be a function. For a code ${C}$, we define $g^*(\mathcal{C}) = \sum_{\mathbf{c} \in C \cap (\mathcal{R}^n)^*} g(c)$. If $\mathcal{C}_b$ is the set of all codes of rank $k$ then
	$$E\left[g^*(C)\right] \leq \frac{\left| \mathcal{R} \right|^k}{\left|(\mathcal{R}^n)^*\right|} g^*(\mathcal{R}^n),$$
	where the expectation is with respect to the uniform distribution on $\mathcal{C}_b$.
	\label{lem:balancedRings}
\end{lemma}
\begin{proof}
	For any balanced set of codes with cardinality $M$, we have
	\begin{equation*}
	\begin{split}
	E[g^*({C})]  &= E\left[\sum_{\mathbf{c} \in {C} \cap (\mathcal{R}^n)^*} g(c)\right] = E\left[\sum_{\mathbf{x} \in (\mathcal{R}^n)^*} g(\xx) \mathbbm{1}_{{C}}(\xx)\right]
	\\ & = \sum_{\mathbf{x} \in (\mathcal{R}^n)^*} E\left[ g(x) \mathbbm{1}_{{C}}(x)\right] \\ &= \sum_{\mathbf{x} \in (\mathcal{R}^n)^*} g(x) \frac{L}{|\mathcal{C}_b|} \leq \frac{M}{|(\mathcal{R}^n)^*|} g^*(\mathcal{R}^n).
	\end{split}
	\end{equation*}
	We now need to prove that the set of all codes of rank $k$ is balanced. Let  $\yy$ be any element in ${{\mathcal{R}^n}^*}$. There exists an invertible linear map $T(\yy) = (1,0,\ldots,0) = \mathbf{e}_1$. Since $T$ is rank-preserving, $\yy \in {C}$ if and only if $\mathbf{e}_1 \in T{({C})}$, where ${C}$ and $T({C})$ have same rank. This induces a bijection between the codes that contain $\yy$ and the codes that contain $e_1$, proving the statement.
\end{proof}

\subsection{Lipschitz and Hurwitz Lattices}
The quaternion skew-field $\mathbb{H}$ is given by $\mathbb{H} = \left\{ a+bi + (c+di)j: a, b, c, d \in \mathbb{R}\right\}$, with the usual relations $i^2=j^2=-1$ and $ij = -ji$. Vance recently \cite{Vance} proved a Minkowski-Hlawka theorem for lattices in $\mathbb{H}$ over the Hurwitz order. Here we show how to recover a ``coding-theoretic'' version of this result from generalized reductions.

We first explain how to deduce a slightly simpler case, for the Lipschitz order. The \textit{Lipschitz integers} $\mathcal{L} \subset \mathbb{H}$ is the (non-maximal) order $\mathcal{L}=\left\{x+yj : x, y \in \mathbb{Z}[i]\right\}.$  Recall that a quaternion has matrix representation
$$\left(\begin{array}{cc} x & -\overline{y} \\ y & \overline{x} \end{array}\right).$$
Let $\mathfrak{p}$ be an ideal in $\mathbb{Z}[i]$ above $p$ that splits. Let $\pi:\mathbb{Z}[i]\to\mathbb{Z}[i]/\mathfrak{p}$ be a projection. We consider the following ``single-letter'' reduction:
\begin{equation*}\phi_p^{\mathbb{H}}:\mathcal{L} \to \mathcal{M}_2(\mathbb{F}_p)
\end{equation*}
\begin{equation*}
\phi_p(x+yj)^{\mathbb{H}}=\left(\begin{array}{cc} \pi(x) & -\pi(\overline{y}) \\ \pi(y) & \pi(\overline{x}) \end{array}\right).
\end{equation*}
We have $\ker \phi_p^{\mathbb{H}} = (p\mathbb{Z}[i])+(p\mathbb{Z}[i])j$. Identifying $\mathbb{H}$ with $\mathbb{R}^4$ in the natural way 
$$\psi(a,b,c,d) \to a+bi + (c+di)j$$
we obtain a reduction $\phi_p:\mathbb{Z}^4 \to \mathcal{M}_2(\mathbb{F}_p)$, $\phi_p(x) = \phi_p^{\mathbb{H}}(\psi(x)).$ By abuse of notation, we will also denote by $\phi_p^{\mathbb{H}}$ the reduction applied componentwise in vector of $\mathcal{L}^m$, i.e., 
\begin{equation}\begin{split} &\phi_p^{\mathbb{H}}(x_1+y_1j,\ldots,x_m+y_mj) \\ &= (\phi_p^{\mathbb{H}}(x_1+y_1j),\ldots,\phi_p^{\mathbb{H}}(x_m+y_mj))\in\mathcal{M}_2(\mathbb{F}_p)^m.
\end{split}
\end{equation}
If ${C} \subset \mathcal{M}_2(\mathbb{F}_p)^m$ is a linear code, then $\Lambda_p^{\mathbb{H}}({C})=(\phi_p^{\mathbb{H}})^{-1}({C})$ is a quaternionic lattice with volume $|C| p^{-4m}$. Let $\mathcal{C}_b$ be a balanced set and $\mathbb{L}_p$ the associated lattice ensembles
$$\mathbb{L}_p = \left\{ \beta \Lambda_p^{\mathbb{H}}({C}) :  {C} \in \mathcal{C}_b \right\},$$
where $\beta = (V/(|C|^{-1} p^{4m}))^{1/(4m)}$. The following Theorem \ref{thm:matrices} is the analogous of Theorem \ref{thm:2} for Lipschitz lattices. We need the following lemma
\begin{lemma}
	If $\phi_p(x+yj)$ is non-invertible for $x+yi\in\mathcal{L}$, then the squared norm of $x+yj$ is a multiple of $p$.
	\label{lem:1}
\end{lemma}
\begin{proof}
	If $\det \phi_p(x+yj) = 0$, then $\pi(x \overline{x} + y \overline{y}) = 0$, i.e., $\left\| (x,y)\right\|^2 \in \mathfrak{p}$. Since the norm of a Lipschitz quaternion is an integer, and $\mathfrak{p}$ is above $p$, the result follows.
\end{proof}
\begin{thm} Let $\mathcal{C}_b$ be a balanced set of codes with rank $k > m/2$. If $f$ is a semi-admissible function then
	\begin{equation}\lim_{p\to\infty} E_{\mathbb{L}_p}\left[\sum_{\xx \in \beta \Lambda_p^{\mathbb{H}}({C})}f(\psi(\xx))\right] \leq (\zeta(4m)V^{-1}) \int_{\mathbb{R}^{4m}} f(\xx) d\xx.
	\label{eq:LipschitzIntegers}
	\end{equation}
	\label{thm:matrices}
\end{thm}
\begin{proof}
	The proof is very similar to that of Theorem \ref{thm:2}. Here we divide the expectation into invertible and non-invertible elements (we make the change of variable $\xx = \beta \yy$, to facilitate), i.e. 
\begin{equation*}\begin{split} \sum_{\yy \in \Lambda_p^{\mathbb{H}}({C})}f(\psi(\beta \yy)) &= \sum_{{\yy \in \Lambda_p^{\mathbb{H}}({C})} \above 0pt {\phi({\yy}) \in (\mathcal{M}_2(\mathbb{F}_p)^m)^{*}}}f(\psi(\beta\yy))\\ &+\sum_{{\xx \in \Lambda_p^{\mathbb{H}}({C})} \above 0pt {\phi({\yy}) \notin (\mathcal{M}_2(\mathbb{F}_p)^m)^{*}}}f(\psi(\beta\yy)).
\end{split}
\end{equation*}
	The first term tends to zero as $p\to\infty$ from Lemma \ref{lem:1}, since $f$ is semi-admissible and
	\begin{equation} \beta\left\| \psi(\yy)\right\| \geq\beta \sqrt{p } = |C|^{1/4m} p^{-1/2} = p^{k/m-1/2} \to \infty,
	\label{eq:boundPHurwitz}
	\end{equation}
	as $p to \infty$.	From Lemma \ref{lem:balancedRings} we conclude that the second term is upper bounded by the right-hand side of \eqref{eq:LipschitzIntegers} as $p \to \infty$.
\end{proof}
For the maximal Hurwitz order 
$$\mathcal{H} = \left\{a+bi+cj+d(-1+i+j+ij)/2 : a, b, c, d \in \mathbb{Z}\right\},$$ the theorem follows by considering reductions from left-prime ideals $\mathcal{P} \triangleleft \mathcal{H}$. For any rational prime $p$, there exist isomorphisms $\mathcal{H}/p\mathcal{H} \sim \mathbb{F}_p(i,j,k) \sim \mathcal{M}_2(\mathbb{F}_p)$ (e.g.  Wedderburn's Theorem \cite[Thm. 6.16 Lem 9.2.1]{Voight}), where non-invertible elements in $\mathcal{H}/p\mathcal{H}$ have reduced norm (determinant) proportional to $p$. Notice that in this case we obtain a reduction
$$\phi_p : D_4 \to \mathcal{M}_2(\mathbb{F}_p),$$
where $D_4$ is the checkerboard lattice in dimension four \cite[Sec. 7.2]{Splag}. An explicit realization of ring isomorphism is obtained by setting
$$\phi(1) = \left(\begin{array}{cc}1 & 0\\0 & 1\end{array}\right), \phi(i) = \left(\begin{array}{cc}0 & -1\\1 & 0\end{array}\right)$$ $$\mbox{ and } \phi(j) = \left(\begin{array}{cc}a & b\\b & -a\end{array}\right),$$
where $a$ and $b$ are two integers such that $a^2+b^2 \equiv -1 \text{ (mod } p)$. Notice that such an isomorphism preserves the residue class of the reduced norm, i.e. $nrd(x) = \det \phi(x) \text{ (mod }p)$, for any $x \in \mathcal{H}$.

\section{Algorithmic Effectiveness}
\label{sec:conc}
Theorem \ref{thm:2} holds in the limit $p_j \to \infty$. However, for each $n$, under some conditions it is possible to find \textit{finite} ensembles that contain dense lattices. In the literature, this is referred to as \textit{effectiveness} (e.g. \cite[p. 18]{Splag} and \cite{Zemor}). We show conditions for a family of reductions to be effective. We need the following lemma, which is a special case of a classical result in the Geometry of Numbers (see \cite[p. 141]{Leker}) and is also valid if $\mathcal{B}_r$ is replaced by more general sets. We include a proof here for the sake of completeness.
\begin{lemma}
	Let $\mathcal{P}$ be a fundamental region for $\Lambda$, an let $l_0 = \sup_{\xx\in\mathcal{P} }\left\|   \xx \right\|.$ For $r > l_0$, we have
	\begin{equation}
	(r-l_0)^n V_n \leq V(\Lambda) {\mathcal{N}_{\Lambda}(r)} \leq (r+l_0)^nV_n.
	\end{equation}
	In particular, we can take $l_0 = \tau(\Lambda)$ to be the covering radius of $\Lambda$.
	\label{lem:pointsEnumerator}
\end{lemma}
\begin{proof}
	We show the set inclusion
	$$\mathcal{B}_{r-l_0} \subset \bigcup_{\xx\in\Lambda\cap \mathcal{B}_{r}}(\xx+\mathcal{P})\subset \mathcal{B}_{r+l_0}.$$
	The lemma then follows from a simple volume calculation of the three sets.
	
	If $\yy = \xx + \mathbf{p}$, $\xx \in \Lambda\cap\mathcal{B}_r$, $\mathbf{p} \in \mathcal{P}$, then $\left\| \yy \right\| \leq \left\| \xx \right\|+\left\| \mathbf{p}\right\|  \leq r+l_0$, proving the second inclusion. For the first inclusion, let $\yy \in \mathcal{B}_{r-l_0}$ and write it as $\yy = \xx + \mathbf{p}$, with $\xx \in \Lambda$ and $\mathbf{p} \in \mathcal{P}$ (this is always true since $\mathcal{P}$ is a fundamental region). Then $\left\| \xx\right\| \leq \left\| \yy \right\| + \left\| \mathbf{p} \right\| \leq r$.
\end{proof}
\subsection{Effective families containing dense lattices}
The next proposition essentially says that if the base lattice $\Lambda$ is sufficiently ``thin'' and the kernel-lattices $\Lambda_{p_j}$ are not so sparse, it is possible to bound $p_j$ in terms of the rate of the underlying code. The conditions are very mild (they are achievable, for instance, by $\Lambda = \mathbb{Z}^m$ and $\Lambda_p = p\mathbb{Z}^m$).  For convenience, we recall the definition of the Hermite parameter $\gamma(\Lambda) = \lambda_1(\Lambda)/V(\Lambda)^{1/m}$ and define the covering parameter as $\mu(\Lambda) = \tau(\Lambda)/V(\Lambda)^{1/m}$. Recall that $\rho(\Lambda) = \lambda_1(\Lambda)/2$ is the packing radius of $\Lambda$. A lattice satisfying the Minkowski-Hlawka bound has density
$$\Delta = \frac{V_m \rho(\Lambda)^m}{V(\Lambda)} > 1/2^{m-1} \Rightarrow \rho(\Lambda) \geq 2^{1-1/m}\left(\frac{V(\Lambda)}{V_m}\right)^{1/m},$$
where $V_m = \vol\, \mathcal{B}_1$ is the volume of the unit ball in $\mathbb{R}^m$. Recalling that $V_m^{-1/m} \sim \sqrt{m/2\pi e}$, if $V(\Lambda)$ is normalized to one, this implies that the packing radius of good lattices should scale as 
$$\rho(\Lambda) \sim  2 \sqrt{\frac{{m}}{2\pi e}}.$$

\begin{prop}
Let the notation be as in Theorem \ref{thm:2} and let $\varepsilon > 0$. Let $\delta = k/n$ be the rate of the underlying codes. Suppose that
\begin{enumerate}
\item[(i)] $p_j^{n\delta/m} \gamma(\Lambda_{p_j}) = \Omega(\sqrt{m})\,\,$ and 
\item[(ii)]  $p_j^{n/m}  = \omega( m \mu(\Lambda)/\gamma(\Lambda_{p_j})).$ 
\end{enumerate}
If $m$ is sufficiently large, there exists a code with parameters $(n,k,p_j)$ such that the lattice $\Lambda_{p_j}(\mathcal{C})$ has packing density greater than $(1-\varepsilon)/2^{m-1}$.
\label{prop:family_sizes}
\end{prop}

\begin{proof} 	For simplicity, suppose that the volume of $\Lambda_{p_j}(\mathcal{C})$ equals $1$, which can be achieved by choosing an appropriate scaling factor. Considering the above discussion let $r = \sqrt{m/2\pi e}$. In the notation of the proof of Theorem \ref{thm:2}, the average lattice point enumerator (Equation \eqref{eq:splitAverage}) becomes:
	\begin{equation*}
	\begin{split} \mathbb{E}\left[\#\left(\beta\Lambda_{p_j}^{\prime}(C)\cap \mathcal{B}_r\right) \right] &= \mathbb{E}\left[\#\left(\beta\Lambda_{p_j}^{\prime}\cap \mathcal{B}_r\right) \right] \\ &+ \mathbb{E}\left[\#\left(\beta(\Lambda_{p_j}^{\prime}(C)\backslash\Lambda_p)\cap \mathcal{B}_r\right)\right]. 
	\end{split}
	\end{equation*}
 The first term of the right-hand side zero whenever 
	\begin{equation} \frac{p_j^{k/m} \lambda_1(\Lambda_{p_j})}{V(\Lambda_{p_j})^{1/m}} \geq r. 
	\label{eq:effectivep}
	\end{equation}
	The second term satisfies
	\begin{equation}
	\begin{split}
	&\mathbb{E}\left[\#\left(\beta(\Lambda_{p_j}^{\prime}(C)\backslash\Lambda_p)\cap \mathcal{B}_r\right)\right] = \frac{p_j^k-1}{p_j^n-1} \#\left(\beta(\Lambda^{\prime}\backslash\Lambda_{p_j})\cap \mathcal{B}_r\right) \\
	\leq & p_{j}^{k-n}(r+\beta \tau(\Lambda))^m \frac{V_m}{\beta^m V(\Lambda)} = V_m r^m \left(1+\frac{\mu(\Lambda)}{r p_j^{(n-k)/m}}\right)^{m}.
	\end{split}
	\label{eq:countPointsEffecitve}
	\end{equation}
Imposing the right-hand-side of \eqref{eq:countPointsEffecitve} to be $2(1-\varepsilon)$, we obtain a lattice with density (cf. Section \ref{sec:prelim}):
\begin{equation} \Delta \geq \frac{1-\varepsilon}{2^{m-1}}\left(1+\frac{\mu(\Lambda)}{r p_j^{(n-k)/m}}\right)^{-m}.
\label{eq:densityEffective}
\end{equation}
Under the conditions of the theorem, the term in parenthesis tends to $1$ as $m \to \infty$.

\end{proof}
\begin{rmk}
	Similar conditions hold for the case of quaternionic lattices. In this case, in light of the proof of Theorem \ref{thm:matrices} and Equation \eqref{eq:boundPHurwitz}, condition (i) should be replaced by $p = \Omega( r^{2m/(2k-m)})$.
\end{rmk}
\begin{exe} Let $m = n$, $\Lambda = \mathbb{Z}^m$ and $\phi_p$ be the ``modulo-$p$'' reduction. Conditions (i)-(ii) of Proposition \ref{prop:family_sizes} state that
$$p \geq c_1 m^{1/2\delta} \mbox{ and } p \geq c_2 m^{3/2+\nu},$$
where $c_1,c_2$ are constants and $\nu$ is any small number. The optimal rate (i.e., the one that yields the smallest asymptotic behavior of $p$) is $\delta \sim 1/3$, which gives us optimal alphabet-size $p = m^{3/2+\nu}$, for any positive constant $\nu$. This provides an alternative derivation of \cite{Rush1989}.

\end{exe}
The alphabet-size in the above example can be further improved by starting the reductions with a lattice which already has a good density, as shown next.

Let $\Lambda = A_{n}^{l}$ be a Craig's lattice \cite[pp.222-224]{Splag} of rank $m = n$, where $n+1=q$ is a prime. From \cite[Prop. 4.1]{Bachoc92}, a Craig's lattice is similar to the embedding of the ideal $(1-\mu_p)^l \mathbb{Z}[\mu_p]$ in the cyclotomic field $\mathbb{Q}(\mu_p)$.  A concrete realization is 
$$\Lambda = \frac{1}{\sqrt{p}} \sigma((1-\mu_p)^l \mathbb{Z}[\mu_p]).$$
From this, we have $\Lambda^{*} \sim A_{n}^{n/2-l}$,
$$\frac{\lambda_1(\Lambda)}{V(\Lambda)^{1/n}} \geq \frac{\sqrt{2l}}{(n-1)^{({2l-1})/2n}}$$ $$\mbox{ and }\frac{\lambda_1(\Lambda^{*})}{V(\Lambda^{*})^{1/n}} \geq \frac{\sqrt{n-2l}}{(n-1)^{({n-2l-1})/2n}}.$$

Following \cite[p. 224]{Splag}'s suggestion, we consider Craig's lattices with parameter $l = \left\lfloor n/2\log (n+1) \right\rceil$ so that, for sufficiently large $n$,
$$\frac{\lambda_1(\Lambda)}{V(\Lambda)^{1/n}} \geq \sqrt{\frac{2\pi}{\log n}}\left(\sqrt{\frac{n}{2\pi e}}+o(1)\right)$$ $$\mbox{ and } \frac{\lambda_1(\Lambda^*)}{V(\Lambda^*)^{1/n}} \geq \sqrt{e}+o(1).$$
From Banaszczyk's transference bound \cite{Banaszczyk}: 
$$\frac{\tau(\Lambda)}{V(\Lambda)^{1/n}} \leq \frac{\sqrt{n}}{2\sqrt{e}+o(1)}.$$
Therefore, using a natural reduction, conditions (i) and (ii) in Proposition \ref{prop:family_sizes} become
$$p \geq c_1 (\sqrt{\log n})^{1/\delta} \mbox{ and } p \geq c_2(n \sqrt{\log n})^{1+\nu},$$
for some constants $c_1,c_2$ and any positive $\nu$. We can further optimize the rate by equalizing the coefficients from where we get
\begin{equation}\delta \sim \frac{\log \log n}{2\log n + \log \log n}.
\label{eq:rate}
\end{equation}

\begin{cor}
	Let $\Lambda = A_n^{\left\lfloor n/2(n+1)\right\rceil}$ and let $\phi_p$ be a natural reduction, as described in Example \ref{ex:naturalReduction}. Let $\varepsilon > 0$, $n$ sufficiently large and let the rate be as in \eqref{eq:rate}. There exists a code $\mathcal{C}$ with parameters 
	$$\left(n,\delta n, O((n \sqrt{\log n})^{1+\nu})\right),$$
	for any positive $\nu$,	such that $\Lambda_p(\mathcal{C})$ has packing density arbitrarily close to $(1-\varepsilon)/2^{m-1}$.
\end{cor}
%
We close this subsection with a comment on the absolute family size of a reduction. If the set of \textit{all} $(n,k,p)$ codes is considered, then the search space for a dense lattice is given by the Gaussian binomial:
\begin{equation}
\left[n \above 0pt k \right]_p = \prod_{i=1}^{k-1} \frac{p^{n}-p^{i}}{p^k-p^i} \sim p^{k(n-k)}.
\end{equation}
Plugging the bounds for $p$ in Proposition \ref{prop:family_sizes} gives an upper bound on the exhaustive search complexity. On Table \ref{tab:comparison} we provide a comparison of the parameters of some constructions in the literature in terms of rank of the base lattice $\Lambda$, code parameters, alphabet-size and log of the family-size
(contrary to a statement in \cite{Zemor,Moustrou}, the complexity of Rush's construction is $\exp(c n^2 \log n)$ rather than $\exp(n\log n)$, and therefore the gains of averaging over double-circulant codes/cyclotomic lattices are even higher than the ones stated).
\renewcommand{\figurename}{Table}
\begin{figure*}
\begin{center}
	\small 
	\begin{tabular}{|c|c|c|c|c|}
		\hline
		Construction & $\text{rank}(\Lambda)$ & $(n,k)$ & $p$ & Log family size \\
		\hline 
		Construction A over $\mathbb{Z}$ \cite{Rush1989,Loeliger} &  $m = n$ & $(n,\delta n), \delta \sim 1/3$ & $n^{3/2}$ & $n^2 \log n$ \\
		\hline 
		Random double-circulant \cite{Zemor} &  $m = n$ & $(n,n/2)$ & $n^{2} \log n$ & $n \log n$ \\
		\hline 
		Cyclotomic lattices \cite{Moustrou} &  $m = 2\Phi(l), l \in \mathbb{N}$ & $(2,1)$ & $l^3(\log l)^{\Phi(l)}$ & $m \log m$ \\
		\hline 
		 \rule{0pt}{0.9\normalbaselineskip} Craig's reduction & $m = n$ & $(n,\delta n), \delta \sim \frac{\log \log n}{2\log n + \log \log n}$ & $n (\log n)^{1/2}$ & $n^2 \log \log n$ \\[1ex] 
		\hline 
	\end{tabular}
	\end{center}
\caption{Parameters of different effective families contaning dense lattices. The rates $\delta$ are up to lower order terms, and the log family-sizes up to constants and lower order terms. $\Phi(l)$ denotes Euler's totient function.}
\label{tab:comparison}
\end{figure*}

\subsection{Packing Efficiency}
A cruder measure of goodness, which is suitable for coding applications, is the \textit{packing efficiency} \cite{ELZ05}. Define $\rho_{\eff}(\Lambda) = (V(\Lambda)/V_m)^{1/m}$ as the radius of a ball whose volume is $V(\Lambda)$. The Minkowski bound can be rephrased in terms of \textit{packing efficiency}, as 
\begin{equation}
\frac{\rho(\Lambda)}{\rho_{\eff}(\Lambda)} \geq \frac{1}{2}.
\end{equation}
A ``packing-good'' family of lattices is such that its packing efficiency is arbitrarily close to $1/2$. As shown in \cite[Sec. IV]{ELZ05}, it is possible to find families with good asymptotic packing efficiency using Loeliger's construction, provided that $p = O(m^{1/2+\beta})$, for any positive small $\beta$. Similarly, we can show that Craig's lattices constructions can achieve packing efficiency arbitrarily close to $1/2$ with alphabet-size $p = O((\log m)^{1/2+\beta})$.

\section{Final Discussion}
\label{sec:conclusion}
\textbf{Applications.} As observed by Loeliger \cite{Loeliger}, random ensembles of lattices are not only good in terms of packing density, but are also sphere-bound achieving when used as infinite constellations for the AWGN channel. Indeed, Rush \cite{Rush1989} and Loeliger's \cite{Loeliger} Construction A $\mathbb{Z}$-lattices are ubiquitous in applications to information transmission over Gaussian channels and networks. However, for other communication problems, such as information transmission in the presence of fading and multiple antennas, it is desirable to enrich the lattices with some algebraic (multiplicative) structure. To this purpose, several recent works such as \cite{Adaptative,KosiOngOggier,ConsAOggier,Our} present different constructions that attach a linear code to an algebraic lattice, but, to date, there is no unified analysis of such ensembles. The generalized reductions described here provide a method for establishing the ``goodness'' of all such constructions at once. It also provides a simple condition to verify if any new construction is ``good'' (e.g., sphere-bound achieving). This was indeed the initial motivation of the author.
\\[1\baselineskip]
\noindent\textbf{Further Perspectives. }The framework considered in this paper is used to provide simple alternative (coding-theoretic) proofs and improvements on previous refinements on the best packing density. It not only implies the existence of dense lattices, but also of \textit{structured lattices}, with the structured inherited from the underlying reduction. 

The question whether it is possible to improve on the $c n\log \log n/2^{n-1}$ asymptotic behavior of cyclotomic fields by specializing the reductions (or the family of codes) appropriately is still open. Furthermore, all known lower bounds on $\Delta_n$ are of the form $\Delta_n \geq 2^{-n(1+\varepsilon(n))}$, with $\varepsilon(n) = O(\log n / n)$, which improves only marginally on the Minkowski-Hlawka lower bound. According to Gruber \cite[p. 388]{GruberDiscrete}, Hlawka believed that no essential improvement can be made, probably meaning that the exponent $2$ is optimal. Nevertheless, the best known upper bound on $\Delta_n$, due to Kabatianskii and Leveshenstein, is of the form $C^{-n}$, where $C \approx 1.51$. Closing this gap is a long-standing open problem.

\section*{Acknowledgments}
The author would like to thank Cong Ling for his constant enthusiasm on the topic and for his suggestions, as well as both reviewers for pointing out inaccuracies in the first version of the manuscript. He also acknowledges Sueli Costa and Jean-Claude Belfiore for fruitful discussions and for hosting him at University of Campinas and at T\'el\'ecom Paristech, where part of this work was developed. This work was supported in part by FAPESP (fellowship 2014/20602-8).


\bibliographystyle{unsrt}
\bibliography{campbel}

\begin{thebibliography}{10}

\bibitem{KosiOngOggier}
W.~Kositwattanarerk, S.~S. Ong, and F.~Oggier.
\newblock {Construction A of Lattices Over Number Fields and Block Fading
  (Wiretap) Coding}.
\newblock {\em IEEE Transactions on Information Theory}, 61(5):2273--2282, May
  2015.

\bibitem{Adaptative}
Y.{-}C. Huang, K.~R. Narayanan, and P.{-}C. Wang.
\newblock Adaptive compute-and-forward with lattice codes over algebraic
  integers.
\newblock arXiv/1501.07740, 2015.

\bibitem{Our}
A.~Campello, C.~Ling, and J.~C. Belfiore.
\newblock Algebraic lattice codes achieve the capacity of the compound
  block-fading channel.
\newblock In {\em IEEE International Symposium on Information Theory (ISIT)},
  pages 910--914, July 2016.

\bibitem{Index}
Y.~C. Huang.
\newblock Lattice index codes from algebraic number fields.
\newblock {\em IEEE Transactions on Information Theory}, 63(4):2098--2112,
  April 2017.

\bibitem{Cassels}
J.~W.~S. Cassels.
\newblock {\em {An Introduction to the {G}eometry of {N}umbers}}.
\newblock Springer-Verlag, 1997.

\bibitem{Leker}
P.~M. Gruber and C.~G. Lekkerkerker.
\newblock {\em {G}eometry of {N}umbers}.
\newblock North-Holland, 1987.

\bibitem{Siegel}
C.~L. Siegel.
\newblock {A Mean Value Theorem in Geometry of Numbers}.
\newblock {\em Annals of Mathematics}, 46(2):340--347, 1945.

\bibitem{Rogers}
C.~A. Rogers.
\newblock {\em {Packing and Covering}}.
\newblock Cambridge University Press, March 1964.

\bibitem{Rush1989}
J.~A. Rush.
\newblock A lower bound on packing density.
\newblock {\em Inventiones mathematicae}, 98(3):499--509, 1989.

\bibitem{Loeliger}
H.-A. Loeliger.
\newblock Averaging bounds for lattices and linear codes.
\newblock {\em IEEE Transactions on Information Theory}, 43(6):1767--1773, Nov
  1997.

\bibitem{Vance}
S.~Vance.
\newblock Improved sphere packing lower bounds from hurwitz lattices.
\newblock {\em Advances in Mathematics}, 227(5):2144 -- 2156, 2011.

\bibitem{Venkatesh}
A.~Venkatesh.
\newblock A note on sphere packings in high dimension.
\newblock {\em International Mathematics Research Notices}, 2012.

\bibitem{Zemor}
P.~Gaborit and G.~Z{\'e}mor.
\newblock On the construction of dense lattices with a given automorphisms
  group.
\newblock {\em Ann. Inst. Fourier (Grenoble)}, 57(4):1051--1062, 2007.

\bibitem{Moustrou}
P.~Moustrou.
\newblock On the density of cyclotomic lattices constructed from codes.
\newblock {\em International Journal of Number Theory}, 13(05):1261--1274,
  2017.

\bibitem{Splag}
J.~H. Conway and N.~J.~A. Sloane.
\newblock {\em {Sphere-packings, lattices, and groups}}.
\newblock Springer-Verlag, New York, NY, USA, 1998.

\bibitem{ConsAOggier}
R.~Vehkalahti, W.~Kositwattanarerk, and F.~Oggier.
\newblock {Constructions A of lattices from number fields and division
  algebras}.
\newblock In {\em IEEE International Symposium on Information Theory (ISIT)},
  pages 2326--2330, June 2014.

\bibitem{LLBS_12}
C.~Ling, L.~Luzzi, J.-C. Belfiore, and D.~Stehle.
\newblock {Semantically Secure Lattice Codes for the Gaussian Wiretap Channel}.
\newblock {\em IEEE Transactions on Information Theory}, 60(10):6399--6416, Oct
  2014.

\bibitem{Fermat}
I.~Stewart and D.~Tall.
\newblock {\em {Algebraic Number Theory and Fermat's Last Theorem}}.
\newblock CRC Press, 3rd edition, 2001.

\bibitem{Notes}
F.~Oggier.
\newblock {\em Introduction to Algebraic Number Theory}.
\newblock Lecture notes available on the author's website
  http://www1.spms.ntu.edu.sg/~frederique/AA11.pdf, accessed on 26/07/2017.

\bibitem{Algebraic1}
J.~Neukirch.
\newblock {\em {Algebraic Number Theory}}, volume 322 of {\em Grundlehren der
  mathematischen Wissenschaften}.
\newblock Springer-Verlag Berlin Heidelberg, 1999.

\bibitem{Autissier}
P.~Autissier.
\newblock {Vari\'et\'es ab\'eliennes et th\'eor\`eme de Minkowski-Hlawka}.
\newblock {\em Manuscripta Mathematica}, 149(3):275--281, 2016.

\bibitem{Voight}
J~Voight.
\newblock {\em The arithmetic of quaternion algebras}.
\newblock 2017, (book available online
  https://math.dartmouth.edu/~jvoight/quat-book.pdf), accessed on 26/07/2017.

\bibitem{Bachoc92}
C.~Bachoc and C.~Batut.
\newblock {\'Etude algorithmique de r\'eseaux construits avec la forme trace}.
\newblock {\em Experiment. Math.}, 1(3):183--190, 1992.

\bibitem{Banaszczyk}
W.~Banaszczyk.
\newblock New bounds in some transference theorems in the geometry of numbers.
\newblock {\em Mathematische Annalen}, 296(1):625--635, 1993.

\bibitem{ELZ05}
U.~Erez, S.~Litsyn, and R.~Zamir.
\newblock Lattices which are good for (almost) everything.
\newblock {\em IEEE Transactions on Information Theory}, 51(10):3401--3416, Oct
  2005.

\bibitem{GruberDiscrete}
P.~Gruber.
\newblock {\em {Convex and Discrete Geometry}}.
\newblock Springer Berlin Heidelberg, 2007.

\end{thebibliography}
\end{document}